\documentclass[english]{IEEEtran}
\usepackage[T1]{fontenc}
\usepackage[latin9]{inputenc}
\usepackage{float}
\usepackage{bm}
\usepackage{amsmath}
\usepackage{graphicx}
\usepackage{amssymb}

\makeatletter

\floatstyle{ruled}
\newfloat{algorithm}{tbp}{loa}
\floatname{algorithm}{Algorithm}

\newtheorem{definition}{Definition}
\newtheorem{remark}{Remark}
\newtheorem{lemma}{Lemma}
\newtheorem{theorem}{Theorem}

\usepackage{bm}
\usepackage{cite}
\usepackage{url}
\usepackage{babel}

\@ifundefined{showcaptionsetup}{}{%
 \PassOptionsToPackage{caption=false}{subfig}}
\usepackage{subfig}
\makeatother

\usepackage{babel}

\begin{document}

\title{Improved Analyses for SP and CoSaMP Algorithms in Terms of Restricted Isometry Constants}
\author{Chao-Bing Song, Shu-Tao Xia, and Xin-ji Liu}
\maketitle
\begin{abstract}
In the context of compressed sensing (CS), both Subspace Pursuit (SP) and Compressive Sampling Matching Pursuit (CoSaMP) are very important iterative greedy recovery algorithms which could reduce the recovery complexity greatly comparing with the well-known $\ell_1$-minimization. Restricted isometry property (RIP) and restricted isometry constant (RIC) of measurement matrices which ensure the convergency of iterative algorithms play key roles for the guarantee of successful reconstructions. In this paper, we show that for the $s$-sparse recovery, the RICs are enlarged to $\delta_{3s}<0.4859$ for SP and
$\delta_{4s}<0.5$ for CoSaMP, which improve the known results significantly. The proposed results also apply to almost sparse signal and corrupted measurements.
 \end{abstract}

\renewcommand{\thefootnote}{\fnsymbol{footnote}} \footnotetext[0]{
This research is supported in part by the Major State Basic Research
Development Program of China (973 Program, 2012CB315803), the National
Natural Science Foundation of China (61371078), and the Research Fund for
the Doctoral Program of Higher Education of China (20100002110033).

All the authors are with the Graduate School at ShenZhen, Tsinghua University, Shenzhen, Guangdong 518055, P.R. China (e-mail: scb12@mails.tsinghua.edu.cn, liuxj11@mails.tsinghua.edu.cn, xiast@sz.tsinghua.edu.cn).} \renewcommand{\thefootnote}{\arabic{footnote}} \setcounter{footnote}{0}

\section{\label{sec:Introduction}Introduction}
As a new paradigm for signals sampling, compressed sensing (CS) \cite{crt,candes2005decoding,donoho2006compressed} has attracted a lot of attention in recent years.
Consider an $s$-sparse signal $\mathbf{x}=(x_{1}, x_{2}, \ldots, x_{N})\in\mathbb{R}^{N}$ which has at most $s$ nonzero entries. Let $\mathbf{\Phi}\in \mathbb{R}^{m\times N}$ be a measurement matrix with $m\ll N$ and $\mathbf{y}=\mathbf{\Phi}\mathbf{x}$ be a measurement vector.
Compressed sensing deals with recovering the original signal $\mathbf{x}$ from the measurement vector $\mathbf{y}$ by finding the sparsest solution to the undetermined linear system $\mathbf{y}=\mathbf{\Phi}\mathbf{x}$, i.e., solving the the following \emph{$\ell_{0}$-minimization} problem:
$\min ||\mathbf{x}||_{0} \; s.t.\; \mathbf{\Phi}\textbf{\textit{x}}=\mathbf{y},
$
where $||\mathbf{x}||_{0}\triangleq |\{i: x_{i}\neq 0\}|$ denotes the $\ell_{0}$-norm of $\mathbf{x}$.
Unfortunately, as a typical combinatorial minimization problem, this optimal recovery algorithm is NP-hard \cite{candes2005decoding}. For this reason, the design of tractable reconstruction algorithms becomes one of the main problems in CS. Lots of algorithms, e.g., convex relaxations, greedy pursuits, etc. \cite{tropp2010computational}, have been proposed to solve the CS problem.
As a convex relaxation to $\ell_0$ minimization, the $\ell_1$ minimization (Basis Pursuit, BP)\cite{candes2005decoding} establishes the foundations of the CS theory for its polynomial recovery complexity, proven recovery guarantee and recovery stability for various kinds of signals. But the recovery complexity $\mathbb{O}(N^3)$ \cite{candes2005decoding,needell2009uniform} of $\ell_1$ minimization is still too high for many practical applications.
As alternatives to $\ell_1$ minimization, the greedy pursuits which are a class of iterative algorithms could reduce the recovery complexity greatly. Moreover, the greedy pursuits often have comparative empirical performance and provide the recovery guarantee described by the well-known restricted isometry property (RIP) \cite{candes2005decoding} originated from the $\ell_1$ minimization.
Restricted isometry constants (RICs) \cite{candes2005decoding} are important to measure the theoretical
guarantee of reconstruction algorithms, and have already become one of the most important theoretical performance indicators. There are many studies on RICs for BP \cite{candes2008restricted,cai2012sharp,cai2013sparse} and greedy pursuits, e.g., orthogonal matching pursuit (OMP) \cite{davenport2010analysis,maleh2011improved,wang2012recovery,zhang2011sparse,foucart2013stability}, iterated hard thresholding (IHT) \cite{Blumensath2009,foucart2012sparse}.

Among representative greedy pursuits, two important and powerful ones come from the subspace pursuit (SP) of Wei and Milenkovic \cite{dai2009subspace} and compressive sampling matching pursuit (CoSaMP) of Needell and Tropp \cite{needell2009cosamp}, both of which are variants of OMP with a little difference. And the empirical performance and proven theoretical guarantees  of the above two algorithms are similar. In their original paper \cite{dai2009subspace}, Dai and Milenkovic obtained $\delta_{3s}<0.205$ to guarantee the SP algorithm to converge with convergence rate $\rho<1$. Lee, Bresler and Junge \cite{lee2012oblique} showed that the RIC for SP can be improved to $\delta_{3s}<0.325$ with $\rho<1$. On the other hand, in their original paper, Needell and Tropp  \cite{needell2009cosamp} gave $\delta_{4s}<0.1$ to guarantee the CoSaMP algorithm to converge with $\rho<1/2$. Foucart \cite{foucart2012sparse} improved the RIC for CoSaMP to
$\delta_{4s}<0.38427 $ with $\rho<1$ and $\delta_{4s}<0.22665$ with $\rho<1/2$.

In this paper, we make a beneficial attempt to improve the theoretical
guarantees for both SP and CoSaMP algorithms. We show that for the $s$-sparse recovery, the RICs are enlarged to $\delta_{3s}<0.4859$ with $\rho<1$ for SP, and $\delta_{4s}<0.5$ with $\rho<1$ and $\delta_{4s}<0.3083$ with $\rho<1/2$ for CoSaMP, which improve the known results significantly. The proposed results also apply to almost sparse signal and corrupted measurements.
The remainder of the paper is organized as follows. Section \ref{sec:priliminaries} introduces the related concepts, lemmas, and algorithmic descriptions of SP and CoSaMP. Section \ref{sec:sp} gives our main results for SP. Then in Section \ref{sec:cosamp}, we give the derivations for CoSaMP which are parallel to ones for SP in Section \ref{sec:sp}. Section \ref{sec:conclusion} concludes the paper with some discussions. Finally, the proofs of some lemmas used in this paper are given in Appendix.

\section{\label{sec:priliminaries}Preliminaries}

Let $\mathbf{x}=(x_1,x_2,\ldots,x_N)\in\mathbb{R}^N$. Let $T\subseteq \{1,2,\ldots,N\}$, and $|T|$ and $\overline{T}$ respectively denote the cardinality and complement of $T$. Let $\mathbf{x}_T\in\mathbb{R}^N$ denote the vector obtained from $\mathbf{x}$ by keeping the $|T|$ entries in $T$ and setting all other entries to zero.  Let $\text{supp}(\mathbf{x})$ denote the support of $\mathbf{x}$ or the set of indices of nonzero entries in $\mathbf{x}$. Note that $\mathbf{x}$ is $s$-sparse if and only if $|\text{supp}(\mathbf{x})|\le s$. For a matrix $\mathbf{\Phi}\in\mathbb{R}^{m\times N}$, let $\mathbf{\Phi}^{\! *}$ denote the (conjugate) transpose of $\mathbf{\Phi}$ and $\mathbf{\Phi}_T$ denote the submatrix that consists of columns of $\mathbf{\Phi}$ with indices in $T$. Let $\mathbf{I}$ denote the identity matrix whose dimension is decided by contexts.

Let $\mathbf{x}_S$ be the best $s$-terms approximation of $\mathbf{x}$, where $|S|=s$ and the set $S$ maintains the indices of the $s$ largest magnitude entries in $\mathbf{x}$.
Consider the general CS model:
\begin{equation}
\mathbf{y}=\mathbf{\Phi}\mathbf{x}+\mathbf{e}=\mathbf{\Phi}\mathbf{x}_S+\mathbf{\Phi}\mathbf{x}_{\overline{S}}+\mathbf{e}=\mathbf{\Phi}\mathbf{x}_S+\mathbf{e}^{\prime},
\label{eq:general_model}
\end{equation}
where $\mathbf{\Phi}\in\mathbb{R}^{m\times N}$ is a measurement matrix with $m\ll N$, $\mathbf{e}\in\mathbb{R}^m$
is an arbitrary noise, $\mathbf{y}\in\mathbb{R}^m$ is a low-dimensional observation, and $\mathbf{e}^{\prime}=\mathbf{\Phi}\mathbf{x}_{\overline{S}}+\mathbf{e}$
denotes the total perturbation by the sparsity defect $\mathbf{x}_{\overline{S}}$ and measurement error $\mathbf{e}$.

From Dai and Milenkovic \cite{dai2009subspace}, Needell and Tropp \cite{needell2009cosamp}, it is known that under a small RIC, both the SP and CoSaMP algorithms can reconstruct $\mathbf{x}$ with bounded mean-square errors (MSE). Moreover, if $\mathbf{x}$ is exactly $s$-sparse and there is no noise, both SP and CoSaMP reconstruct $\mathbf{x}$ perfectly. The two algorithms are described as follows.
\begin{algorithm}[H]
Input: $\mathbf{y},\mathbf{\Phi},s$.\\
Initialization: $S^{0}=\emptyset,\mathbf{x}^0=\mathbf{0}$.\\
Iteration: At the $n$-th iteration, go through the following steps.
\begin{enumerate}
\item $\Delta S =$ \{$s$ indices corresponding to the  $s$ largest magnitude entries in the vector $\mathbf{\Phi}^{\! *}\, (\mathbf{y}- \mathbf{\Phi} \mathbf{x}^{n-1})$\}.
\item $\tilde{S}^{n}=S^{n-1} \bigcup \Delta S$.
\item ${\mathbf{\tilde{x}}}^{n}=\text{arg}\min_{\mathbf{z} \in\mathbb{R}^N}\{ \Vert\mathbf{y}-\mathbf{\Phi}\mathbf{z}\Vert_2,\;
    \text{supp}(\mathbf{z})\subseteq \tilde{S}^{n}\}$.
\item $S^{n}=$\{$s$ indices corresponding to the  $s$ largest magnitude elements of $\mathbf{\tilde{x}}^{n}$\}.
\item $\mathbf{x}^{n}=\text{arg}\min_{\mathbf{z} \in\mathbb{R}^N}\{ \Vert\mathbf{y}-\mathbf{\Phi}\mathbf{z}\Vert_2,\;
    \text{supp}(\mathbf{z})\subseteq S^{n}\}$.
\end{enumerate}
until the stopping criteria is met. \\
Output: $\mathbf{x}^{n}$, \text{supp}($\mathbf{x}^{n}$).
\caption{Subspace Pursuit}
\end{algorithm}
\begin{algorithm}[H]
Input: $\mathbf{y}$, $\mathbf{\Phi}$, $s$.\\
Initialization: $S^{0}=\emptyset,\mathbf{x}^0=\mathbf{0}$.\\
Iteration: At the $n$-th iteration, go through the following steps.
\begin{enumerate}
\item $\Delta S =$ \{$2s$ indices corresponding to the  $2s$ largest magnitude entries in the vector $\mathbf{\Phi}^{\! *}\, (\mathbf{y}- \mathbf{\Phi} \mathbf{x}^{n-1})$\}.
\item $\tilde{S}^{n}=S^{n-1} \bigcup\Delta S$.
\item ${\mathbf{\tilde{x}}}^{n}=\text{arg}\min_{\mathbf{z} \in\mathbb{R}^N}\{ \Vert\mathbf{y}-\mathbf{\Phi}\mathbf{z}\Vert_2,\;
    \text{supp}(\mathbf{z})\subseteq \tilde{S}^{n}\}$.
\item $S^{n}=$\{$s$ indices corresponding to the $s$ largest magnitude elements of $\mathbf{\tilde{x}}^{n}$\}.
\item $\mathbf{x}^{n}=$ \{the vector from ${\mathbf{\tilde{x}}}^{n}$ that keeps the entries of ${\mathbf{\tilde{x}}}^{n}$ in $S^{n}$ and set all other ones to zero.\}
\end{enumerate}
until the stopping criteria is met. \\
Output: $\mathbf{x}^{n}$, $\text{supp}(\mathbf{x}^{n})$.
\caption{Compressive Sampling Matching Pursuit}
\end{algorithm}

For both the SP and CoSaMP algorithms in the above descriptions, similar to Needell and Tropp \cite{needell2009cosamp}, we call the steps 1 and 2 `identification', and the step 4 `pruning'. In the step 3 of both algorithms, a least squares process is used for debiasing. In the step 5, SP solves a least squares problem again to get the final approximation of the current iteration, while CoSaMP directly keeps the $s$ largest magnitude entries of $\tilde{\mathbf{x}}^{n}$. The stopping criteria can be selected according to the property of algorithm or the need in practice. A stopping criterion for both algorithms could be ``$\Vert\mathbf{y}- \mathbf{\Phi} \mathbf{x}^{n}\Vert_2\le\varepsilon \Vert\mathbf{e}^{\prime}\Vert_2$ or $n \ge n_{\max}$''. Under the prior knowledge that the algorithm will converge with bounded MSE, such a stopping criterion with proper arguments $\varepsilon, n_{\max}$ provides the tradeoff between recovery accuracy and time complexity.

The definitions of RIP and RIC are given in \cite{candes2005decoding} as follows.
\begin{definition}[\cite{candes2005decoding}]
\label{def:rip}
The measurement matrix $\mathbf{\Phi}\in\mathbb{R}^{m\times N}$ is said to satisfy the $s$-order RIP if for any $s$-sparse signal $\mathbf{x}\in\mathbb{R}^{N}$
\begin{equation}
(1-\delta)\Vert\mathbf{x}\Vert_2^2\le\Vert\mathbf{\Phi}\mathbf{x}\Vert_2^2\le(1+\delta)\Vert\mathbf{x}\Vert_2^2, \label{eq:origin_def}
\end{equation}
where $0\le\delta\le1$. The infimum of $\delta$, denoted by $\delta_s$, is called the RIC of $\mathbf{\Phi}$.
\end{definition}

Foucart \cite{foucart2012sparse} pointed that the RIC $\delta_s$ could be formulated equivalently as
\begin{equation}
\delta_s=\max_{S\subseteq\{1,2,\ldots,N\},|S|\le s} \Vert\mathbf{\Phi}_S^{\! *}\mathbf{\Phi}_S-\mathbf{I}\Vert_{2\rightarrow 2},\label{eq:new_def}
\end{equation}
where
\begin{equation}
\Vert\mathbf{\Phi}_S^{\! *}\mathbf{\Phi}_S-\mathbf{I}\Vert_{2\rightarrow 2}=\sup_{\mathbf{a}\in\mathbb{R}^{|S|}\backslash\{\mathbf{0}\}}\dfrac{\Vert(\mathbf{\Phi}_S^{\! *}\mathbf{\Phi}_S-\mathbf{I})\mathbf{a}\Vert_2}{\Vert\mathbf{a}\Vert_2}.\label{eq:new_def2}
\end{equation}

Throughout the paper, we use the notation $(i)$ stacked over an inequality sign to indicate
that the inequality follows from the expression $(i)$ in the paper. The following two lemmas are frequently used in the derivations of RIC related results. For completeness, we include the proofs in Appendix \ref{sec:proof-of-rip}.
\begin{lemma}[Consequences of the RIP \cite{candes2005decoding,foucart2011hard}]$ $
\label{lem:rip}
\begin{enumerate}
\item (Monotonicity) For any two positive integers $s\le s^{\prime}$, $$\delta_s\le\delta_{s^{\prime}}.$$
\item For two vectors $\mathbf{u}, \mathbf{v}\in\mathbb{R}^{N}$, if $|$supp($\mathbf{u}$)$\cup$supp($\mathbf{v}$)$|$$\le t$, then
\begin{eqnarray}
        |\langle\mathbf{u}, (\mathbf{I}-\mathbf{\Phi}^{\! *}\mathbf{\Phi})\mathbf{v}\rangle|\le\delta_t\Vert\mathbf{u}\Vert_2 \Vert\mathbf{v}\Vert_2;\label{rip11}
\end{eqnarray}
      moreover, if $U\subseteq\{1,\dots,N\}$ and $|U \cup \text{supp}(\mathbf{v})$$|$$\le t$, then
\begin{eqnarray}
            \Vert((\mathbf{I}-\mathbf{\Phi}^{\! *}\mathbf{\Phi})\mathbf{v})_U\Vert_2\le\delta_t\Vert\mathbf{v}\Vert_2.\label{rip12}
\end{eqnarray}
\end{enumerate}
\end{lemma}
\vspace{0.1in}
\begin{lemma}[Noise perturbation in partial supports \cite{foucart2011hard}]
\label{lem:noise}
For the general CS model $
\mathbf{y}=\mathbf{\Phi}\mathbf{x}_S+\mathbf{e}^{\prime}$ in (\ref{eq:general_model}), letting $U\subseteq\{1,\ldots,N\}$ and $|U|\le u$, we have
\begin{eqnarray}
\label{rip13}
\Vert(\mathbf{\Phi}^{\! *}\mathbf{e}^{\prime})_{U}\Vert_2\le\sqrt{1+\delta_{u}}\Vert\mathbf{e}^{\prime}\Vert_2.
\end{eqnarray}
\end{lemma}
\begin{proof}
The lemma easily follows from the fact that
\begin{eqnarray*}
&&\Vert(\mathbf{\Phi}^{\! *}\mathbf{e}^{\prime})_{U}\Vert_2^2\\
&=&\langle\mathbf{\Phi}^{\! *}\mathbf{e}^{\prime},(\mathbf{\Phi}^{\! *}\mathbf{e}^{\prime})_{U}\rangle \\
&=&\langle\mathbf{e}^{\prime},\mathbf{\Phi}((\mathbf{\Phi}^{\! *}\mathbf{e}^{\prime})_{U})\rangle \\ &\le&\Vert\mathbf{e}^{\prime}\Vert_2\Vert\mathbf{\Phi}((\mathbf{\Phi}^{\! *}\mathbf{e}^{\prime})_{U})\Vert_2\\
&\overset{(\ref{eq:origin_def})}{\le}&\Vert\mathbf{e}^{\prime}\Vert_2\sqrt{1+\delta_{u}}\Vert(\mathbf{\Phi}^{\! *}\mathbf{e}^{\prime})_U\Vert_2.
\end{eqnarray*}
\end{proof}

The next lemma introduces a simple inequality which is useful in our derivations.
\begin{lemma}
\label{lem:cauchy}
For nonnegative numbers $a,b,c,d,x,y$,
\begin{align}
\label{lem3}
(ax+by)^2+(cx+dy)^2\le(\sqrt{a^2+c^2}x+(b+d)y)^2.
\end{align}
\end{lemma}
\begin{proof}
By the well-known Cauchy inequality $ab+cd\le\sqrt{a^2+c^2}\sqrt{b^2+d^2}$ and $b^2+d^2\le(b+d)^2$,
\begin{eqnarray*}
&&(ax+by)^2+(cx+dy)^2\\
&=&(a^2+c^2)x^2+2(ab+cd)xy+(b^2+d^2)y^2 \\
&\le&(a^2+c^2)x^2+2\sqrt{a^2+c^2}\sqrt{b^2+d^2}xy+(b^2+d^2)y^2\\
&\le&(a^2+c^2)x^2+2\sqrt{a^2+c^2}(b+d)xy+(b+d)^2 y^2\\
&=&(\sqrt{a^2+c^2}x+(b+d)y)^2.
\end{eqnarray*}
\end{proof}

\section{\label{sec:sp}Subspace Pursuit}

Before the detailed derivations, we should note that the framework of our proofs  mainly follows the analysis of Foucart \cite{foucart2012sparse}, \cite{foucart2011hard}. The main differences are that
we take advantage of orthogonality property in Lemma \ref{lem:orthogonality} to get Lemma \ref{lem:identification-sp} in steps 1 and 2 of SP, meanwhile, we obtain new properties of the least squares problems in Lemma 5 which could unify the derivations for both SP and CoSaMP.

Consider the general CS model $
\mathbf{y}=\mathbf{\Phi}\mathbf{x}_S+\mathbf{e}^{\prime}$ in (\ref{eq:general_model}). Let $T\subseteq \{1,2,\ldots,N\}$ and $|T|=t$. Let $\mathbf{z}_{p}$ be the solution of the least squares problem $\mbox{arg}\min_{\mathbf{z}\in\mathbb{R}^{N}}\{\Vert\mathbf{y}-
\mathbf{\Phi}\mathbf{z}\Vert_2, \;\text{supp}(\mathbf{z}$)$\subseteq T\}$. Solving a least squares problem is the same step in both the SP and CoSaMP algorithms. It has the following orthogonal properties.
\begin{lemma}[Basic observations for orthogonality]
\label{lem:orthogonality}
$$(\mathbf{\Phi}^{\! *}(\mathbf{y}-\mathbf{\Phi}\mathbf{z}_{p}))_{T}=\mathbf{0}.$$
\end{lemma}
\begin{proof}
Due to the orthogonality, the residue $\mathbf{y}-\mathbf{\Phi}\mathbf{z}_{p}$ is orthogonal to the space $\{\mathbf{\Phi}\mathbf{z},\text{supp}(\mathbf{z})\subseteq T\}$. This means that for all $\mathbf{z}\in\mathbb{R}^{N}$ with $\text{supp}(\mathbf{z})\subseteq T$,
\begin{eqnarray}
\label{t1}
\langle\mathbf{y}-\mathbf{\Phi}\mathbf{z}_{p},\mathbf{\Phi}\mathbf{z}\rangle=\langle\mathbf{\Phi}^{\! *}(\mathbf{y}-\mathbf{\Phi}\mathbf{z}_{p}),\mathbf{z}\rangle=0,
\end{eqnarray}
which implies the conclusion.
\end{proof}

\begin{remark}
\label{rem2}
By substituting $\mathbf{y}=\mathbf{\Phi}\mathbf{x}_S+\mathbf{e}^{\prime}$  into (\ref{t1}), we have
\begin{eqnarray*}
0&=&\langle\mathbf{\Phi}\mathbf{x}_{S}+\mathbf{e}^{\prime}-
\mathbf{\Phi}\mathbf{z}_{p},\mathbf{\Phi}\mathbf{z}\rangle\\
&=&\langle\mathbf{\Phi}(\mathbf{x}_{S}-\mathbf{z}_{p}),\mathbf{\Phi}
\mathbf{z}\rangle+\langle\mathbf{e}^{\prime},\mathbf{\Phi}\mathbf{z}\rangle\\
&=&\langle\mathbf{x}_{S}-\mathbf{z}_{p},\mathbf{\Phi}^{\! *}\mathbf{\Phi}\mathbf{z}\rangle+
\langle\mathbf{e}^{\prime},\mathbf{\Phi}\mathbf{z}\rangle.
\end{eqnarray*}
Hence, we have that for all $\mathbf{z}\in\mathbb{R}^{N}$ with $\text{supp}(\mathbf{z})\subseteq T$,
\begin{equation}
\label{t2}
\langle\mathbf{x}_{S}-\mathbf{z}_{p},\mathbf{\Phi}^{\! *}\mathbf{\Phi}\mathbf{z}\rangle+\langle\mathbf{e}^{\prime},
\mathbf{\Phi}\mathbf{z}\rangle=0,
\end{equation}
which will be used in subsequent derivations.
\end{remark}

The next lemma is crucial to get the main results, where the proof is referred to Appendix \ref{sec:proof-of-orghononality-rip}.
\begin{lemma}[Consequences for orthogonality by the RIP]
\label{lem:orthogonality-rip}
If $\delta_{s+t}<1$,
\begin{equation}
\Vert(\mathbf{x}_S-\mathbf{z}_{p})_T\Vert_2\le\delta_{s+t}\Vert\mathbf{x}_S-\mathbf{z}_{p}\Vert_2+\sqrt{1+\delta_{t}}\Vert\mathbf{e}^{\prime}\Vert_2\label{eq:orthogonality-rip1}
\end{equation}
and
\begin{equation}
\Vert\mathbf{x}_S-\mathbf{z}_{p}\Vert_2\le\sqrt{\dfrac{1}{1-\delta_{s+t}^2}}
\Vert(\mathbf{x}_S)_{\overline{T}}\Vert_2+\dfrac{\sqrt{1+\delta_{t}}}
{1-\delta_{s+t}}\Vert\mathbf{e}^{\prime}\Vert_2.\label{eq:orthogonality-rip2}
\end{equation}
Moveover, if $t>s$, define $T_{\nabla}$:=\{The indices of the $t-s$ smallest magnitude entries of $\mathbf{z}_{p}$ in $T$\}, we have
\begin{eqnarray}
\!\!\!\!\!\!\Vert(\mathbf{x}_S)_{T_{\nabla}}\Vert_2\!\!\!\!&\le&\!\!\!\!\!\!\sqrt{2}\Vert(\mathbf{x}_S
-\mathbf{z}_{p})_T\Vert_2\nonumber\\
\!\!\!\!\!\!\!\!&\le&\!\!\!\!\!\!\sqrt{2}\delta_{s+t}\Vert\mathbf{x}_S
-\mathbf{z}_{p}\Vert_2+\sqrt{2(1+\delta_{t})}\Vert\mathbf{e}^{\prime}\Vert_2.
\label{eq:orthogonality-rip3}
\end{eqnarray}
\end{lemma}

\begin{remark}
\label{rmk:orthogonality-rip}
Consider the exact reconstruction circumstance, i.e., $\mathbf{x}$ is exactly $s$-sparse and $\mathbf{e}=\mathbf{0}$, which implies that $\Vert\mathbf{e}^{\prime}\Vert_2=0$. Because the RIC is deemed to be small, the left-hand of (\ref{eq:orthogonality-rip1}) is smaller than $\Vert\mathbf{x}_S-\mathbf{z}_{p}\Vert_2$ , which implies that the approximation vector $\mathbf{z}_{p}$  has good approximation effect in supports $T$, so verifies the debiasing effect of the least squares process. Similarly, under a small RIC, from (\ref{eq:orthogonality-rip3}), we know that the signal energy $\Vert\mathbf{x}\Vert_2$ is small in $T_{\nabla}$, which implies that the pruning process in both SP and CoSaMP will not bring in big errors. The inequality (\ref{eq:orthogonality-rip2}) shows that the residual energy $\Vert\mathbf{x}_S-\mathbf{z}_{p}\Vert_2$ is bounded by the signal energy $\Vert(\mathbf{x}_S)_{\overline{T}}\Vert_2$ that falls in $\overline{T}$. From \cite{Blumensath2009,dai2009subspace,needell2009cosamp}, we know that two kinds of $\ell_2$ norm can be used to prove the convergency of greedy pursuits, which are called \emph{approximation metrics} in this paper. Dai and Milenkovic  \cite{dai2009subspace} uses $\Vert(\mathbf{x}_S)_{\overline{T}}\Vert_2$, while Needell and Tropp \cite{needell2009cosamp} and Blumensath and Davies \cite{Blumensath2009} use $\Vert\mathbf{x}_S-\mathbf{z}_{p}\Vert_2$. The inequality (\ref{eq:orthogonality-rip2}) reveals the relation between the two approximation metrics. The inequality (\ref{eq:orthogonality-rip3}) helps us acquire
more general bounds for both SP and CoSaMP algorithms.
\end{remark}

\begin{remark}
\label{rem22}
In steps 1 and 2 of both SP and CoSaMP, if $\Delta S\cap S^{n-1}\neq\emptyset$, which implies $|\tilde{S}^{n}|<2s,|\tilde{S}^{n}\backslash S^{n}|<s$ for SP, and $|\tilde{S}^{n}|<3s,|\tilde{S}^{n}\backslash S^{n}|<2s$ for CoSaMP, by the monotonicity of the RIP in Lemma \ref{lem:rip}, this will not affect the subsequent derivations when using Lemma \ref{lem:orthogonality-rip}. Hence, we assume that $$\Delta S \cap S^{n-1}=\emptyset$$ without loss of generality in the following parts of this paper.
\end{remark}

\vspace{0.1in}
The  steps 1 and 2  of both SP and CoSaMP update the current  estimate of support set by greedily adding indices of some largest magnitude entries of the one-dimensional approximation of $\mathbf{x}_{S}-\mathbf{x}^{n-1}$ to the existing estimate  $S^{n-1}$. We call the process of update `identification'. In the identification step, we have the following lemma for SP
whose proof is postponed to Appendix \ref{sec:proof-of-identification-sp}, while a similar lemma for CoSaMP is proposed in next section.
\begin{lemma}[Identification for SP] In the steps 1 and 2 of SP, we have
\label{lem:identification-sp}
\[\Vert(\mathbf{x}_S)_{\overline{\tilde{S}^{n}}}\Vert_2\le\sqrt{2}\delta_{3s}\Vert\mathbf{x}_S- \mathbf{x}^{n-1}\Vert_2+\sqrt{2(1+\delta_{2s})}\Vert\mathbf{e}^{\prime}\Vert_2.
\]
\end{lemma}

\bigskip
Then we give our main result for SP.
\begin{theorem}
\label{thm:main_result}
For the general CS model $
\mathbf{y}=\mathbf{\Phi}\mathbf{x}_S+\mathbf{e}^{\prime}$ in (\ref{eq:general_model}), if $\delta_{3s}<0.4859$, then the sequence of $\mathbf{x}^{n}$ defined by SP satisfies
\begin{equation}
\label{thm-eq}
\Vert\mathbf{x}_S-\mathbf{x}^{n}\Vert_2\le\rho^{n}\Vert\mathbf{x}_S\Vert_2
+\tau\Vert\mathbf{e}^{\prime}\Vert_2,
\end{equation}
where
\begin{eqnarray}
\rho&=&\dfrac{\sqrt{2\delta_{3s}^2(1+\delta_{3s}^2)}}{1-\delta_{3s}^2}\;<\;1,  \label{eq:rho}
\end{eqnarray}
\begin{eqnarray}
(1-\rho)\tau&=&\sqrt{\dfrac{2\delta_{3s}^2}{1-\delta_{3s}^2}}\left(\dfrac{\sqrt{2(1-\delta_{3s})}+\sqrt{1+\delta_{3s}}}{1-\delta_{3s}}\right)\nonumber\\
&&\quad+\dfrac{2\sqrt{2(1-\delta_{3s})}+\sqrt{1+\delta_{3s}}}{1-\delta_{3s}}.
\label{eq:tau}
\end{eqnarray}
\end{theorem}
\medskip
\begin{proof}
The steps 1 and 2 of SP are the identification steps. By Remark \ref{rem22}, we assume that $|\tilde S^n|=2s$ and $\Delta S\cap S^{n-1} =\emptyset$ without loss of generality. By Lemma \ref{lem:identification-sp}, in the $n$-th iteration, we have
\begin{eqnarray}
\!\!\Vert(\mathbf{x}_S)_{\overline{\tilde{S}^{n}}}\Vert_2\le\sqrt{2}\delta_{3s}\Vert\mathbf{x}_S- \mathbf{x}^{n-1}\Vert_2+\!\!\sqrt{2(1+\delta_{2s})}\Vert\mathbf{e}^{\prime}\Vert_2.\label{eq:sp_1}
\end{eqnarray}
The step 3 of the $n$-th iteration is a procedure of solving a least squares problem. Letting $T=\tilde{S}^{n}$ and $\mathbf{z}_p=\tilde{\mathbf{x}}^{n}$, $t=2s$, by (\ref{eq:orthogonality-rip2}) of Lemma \ref{lem:orthogonality-rip}, we have
\begin{eqnarray}
\!\!\Vert\mathbf{x}_S-\tilde{\mathbf{x}}^{n}\Vert_2\le\sqrt{\dfrac{1}{1-\delta_{3s}^2}}\Vert(\mathbf{x}_S)_{\overline{\tilde{S}^{n}}}\Vert_2+\dfrac{\sqrt{1+\delta_{2s}}}{1-\delta_{3s}}\Vert\mathbf{e}^{\prime}\Vert_2.\label{eq:sp_2}
\end{eqnarray}
Then combining (\ref{eq:sp_1}) and (\ref{eq:sp_2}) and magnifying $\delta_{2s}$ to $\delta_{3s}$ by Lemma \ref{lem:rip}, we have
\begin{eqnarray}
\Vert\mathbf{x}_S-\mathbf{\tilde{x}}^{n}\Vert_2
\!\!&\le&\!\!\sqrt{\dfrac{2\delta_{3s}^2}{1-\delta_{3s}^2}}\Vert\mathbf{x}_S-\mathbf{x}^{n-1}\Vert_2\nonumber\\
\!\!&&\!\!+\dfrac{\sqrt{2(1-\delta_{3s})}+\sqrt{1+\delta_{3s}}}{1-\delta_{3s}}\Vert\mathbf{e}^{\prime}\Vert_2.\label{eq:sp_3}
\end{eqnarray}
After the step 4 of the $n$-th iteration, define $S_{\nabla}:=\tilde{S}^{n}\backslash S^{n}$, where $S_{\nabla}$ contains the indices of the $s$ smallest entries in $\mathbf{\tilde{x}}^{n}$. Letting $T=\tilde{S}^{n}$ and $\mathbf{z}_p=\tilde{\mathbf{x}}^{n}$, $t=2s$, $T_{\nabla}=S_{\nabla}$, by (\ref{eq:orthogonality-rip3}) of Lemma \ref{lem:orthogonality-rip}, we have that
\begin{eqnarray}
\Vert(\mathbf{x}_S)_{S_{\nabla}}\Vert_2\le \sqrt{2}\delta_{3s}\Vert\mathbf{x}_S-\tilde{\mathbf{x}}^{n}\Vert_2+\sqrt{2(1+\delta_{2s})}\Vert\mathbf{e}^{\prime}\Vert_2.\label{eq:sp_4}
\end{eqnarray}
Let $\tau_1=\dfrac{\sqrt{2(1-\delta_{3s})}+\sqrt{1+\delta_{3s}}}{1-\delta_{3s}}$ and $\tau_2=\sqrt{1+\delta_{3s}}$. \\
Dividing $\overline{S^{n}}$ into two disjoint parts: $S_{\nabla}$ and $\overline{\tilde{S}^{n}}$, we have
\begin{eqnarray*}
\!\!\!\!\!\!&&\!\!\!\!\!\!\Vert(\mathbf{x}_S)_{\overline{S^{n}}}\Vert_2^2
=\Vert(\mathbf{x}_S)_{S_{\nabla}}\Vert_2^2+\Vert(\mathbf{x}_S)_{\overline{\tilde{S}^{n}}}\Vert_2^2\\
\!\!\!\!\!\!\!\!\!\!\!\!\!\!&\overset{(\ref{eq:sp_4}), (\ref{eq:sp_1})}{\le}&\!\!\!\!\!\!2\left(\delta_{3s}\Vert\mathbf{x}_{S}-
\tilde{\mathbf{x}}^{n}\Vert_2
+\tau_{2}\Vert\mathbf{e}^{\prime}\Vert_2\right)^2\\
\!\!\!\!\!\!&&\quad+2\left(\delta_{3s}\Vert\mathbf{x}_{S}-\mathbf{x}^{n-1}\Vert_2+\tau_{2}\Vert\mathbf{e}^{\prime}\Vert_2\right)^2\\
\!\!\!\!\!\!\!\!\!\!\!\!&\overset{(\ref{eq:sp_3})}{\le}&\!\!\!\!\!\!\!\!\!\!2\left(\delta_{3s}\sqrt{\dfrac{2\delta_{3s}^2}{1-\delta_{3s}^2}}\Vert\mathbf{x}_S
-\mathbf{x}^{n-1}\Vert_2+(\delta_{3s}\tau_1\!+\!\tau_2)
\Vert\mathbf{e}^{\prime}\Vert_2\!\!\right)^{\!\!2}\\
\!\!\!\!\!\!&&\quad+2\left(\delta_{3s}\Vert\mathbf{x}_S-\mathbf{x}^{n-1}\Vert_2+\tau_2\Vert\mathbf{e}^{\prime}\Vert_2\right)^2
\nonumber\\
\!\!\!\!\!\!\!\!\!\!\!\!&\overset{(\ref{lem3})}{\le}&\!\!\!\!\!\!\!\!\!\!
2\Bigg(\sqrt{\dfrac{2\delta_{3s}^4}{1-\delta_{3s}^2}+\delta_{3s}^2}\;
\Vert\mathbf{x}_S-\mathbf{x}^{n-1}\Vert_2\\
\!\!\!\!\!\!&&\quad+\left((\delta_{3s}\tau_1+\tau_2)+\tau_2\right)
\Vert\mathbf{e}^{\prime}\Vert_2\Bigg)^2\nonumber\\
\!\!\!\!\!\!\!\!\!\!\!\!&=&\!\!\!\!\!\!\!\!\!\!2\left(\!\!\sqrt{\dfrac{\delta_{3s}^2 (1+\delta_{3s}^2)}{1-\delta_{3s}^2}}\Vert\mathbf{x}_S\!-\!\mathbf{x}^{n-1}\Vert_2
\!+\!(\delta_{3s}\tau_1\!+\!2\tau_2)\Vert\mathbf{e}^{\prime}\Vert_2\!\!\right)^{\!\!2},\nonumber
\end{eqnarray*}
which implies that
\begin{eqnarray}
\Vert(\mathbf{x}_{S})_{\overline{S^{n}}}\Vert_2
&\le&\sqrt{\dfrac{2\delta_{3s}^2 (1+\delta_{3s}^2)}{1-\delta_{3s}^2}}\Vert\mathbf{x}_S-\mathbf{x}^{n-1}\Vert_2\nonumber\\
&&\qquad+\sqrt{2}(\delta_{3s}\tau_1+2\tau_2)\Vert\mathbf{e}^{\prime}\Vert_2.\label{eq:sp_5}
\end{eqnarray}
The step 5 of the $n$-th iteration also solves a least squares problem. Letting $T=S^{n}$ and $\mathbf{z}_p=\mathbf{x}^{n}$, $t=s$, by (\ref{eq:orthogonality-rip2}) of Lemma \ref{lem:orthogonality-rip}, we have
\begin{align}
&\Vert\mathbf{x}_S-\mathbf{x}^{n}\Vert_2\le\sqrt{\dfrac{1}{1-\delta_{2s}^2}}\Vert(\mathbf{x}_S)_{\overline{S^{n}}}\Vert_2+\dfrac{\sqrt{1+\delta_{s}}}{1-\delta_{2s}}\Vert\mathbf{e}^{\prime}\Vert_2.\label{eq:sp_6}
\end{align}
Hence, by combining (\ref{eq:sp_5}) and (\ref{eq:sp_6}), and magnifying $\delta_{s},\delta_{2s}$ to $\delta_{3s}$, it is easy to obtain that
\begin{align*}
\Vert\mathbf{x}_S-\mathbf{x}^{n}\Vert_2\le\rho\Vert\mathbf{x}_S
-\mathbf{x}^{n-1}\Vert_2+(1-\rho)\tau\Vert\mathbf{e}^{\prime}\Vert_2,
\end{align*}
where $\rho$ and $\tau$ is respectively referred to (\ref{eq:rho}) and (\ref{eq:tau}). Hence, (\ref{thm-eq}) follows by recursively using the above inequality when $\rho<1$.
Note that $\rho<1$ if \[
\delta_{3s}^4+4\delta_{3s}^2-1<0\qquad \mbox{or} \quad\delta_{3s}<\sqrt{\sqrt{5}-2}\approx 0.4859,\]
which finishes the proof.
\end{proof}

To the best of our knowledge, the previously best-known theoretical result for SP is referred to  Lee, Bresler and  Junge \cite[Theorem 2.9]{lee2012oblique}, where $\delta_{3s}<0.325$ is given to guarantee the convergence of SP with $\rho<1$. We rewrite the corresponding result in \cite[Theorem 2.9]{lee2012oblique} as follows\footnote{In order to compare the error coefficient numerically, we magnify $\delta_{s},\delta_{2s}$ to $\delta_{3s}$ in the original expression of  \cite[Theorem 2.9]{lee2012oblique} with similarity to our derivations.}
\begin{equation}
\Vert\mathbf{x}_S-\mathbf{x}^{n}\Vert_2\le\bar{\rho}
\Vert\mathbf{x}_S-\mathbf{x}^{n-1}\Vert_2+(1-\bar{\rho})\bar{\tau}
\Vert\mathbf{e}^{\prime}\Vert_2,
\end{equation}
where
\begin{eqnarray}
&&\;\bar{\rho}\;=\; \dfrac{\delta_{3s}\sqrt{1+\delta_{3s}}}
{\sqrt{1-\delta_{3s}}}\max\!\!\left\{\!\!\dfrac{1}{(1\!-\!\delta_{3s})^2},
\dfrac{2}{1\!+\!2\delta_{3s}\!+\!2\delta_{3s}^2}\!\right\}\!,\nonumber\\
&&\!\!\!\!\!\!\!\!\!\!\!\!\!\!\!\!(1-\bar{\rho})\bar{\tau}\;=\;
\dfrac{\sqrt{1+\delta_{3s}}}{1-\delta_{3s}}
+\dfrac{1}{\sqrt{1-\delta_{3s}}(1-\delta_{3s})}\nonumber\\
&&\!\!\!\!\!\!\!\!+\dfrac{2(1\!+\!\delta_{3s})^2}
{\sqrt{1\!-\!\delta_{3s}}\;(1\!-\!\delta_{3s})}
\max\!\!\left\{\!\!\dfrac{1}{(1\!-\!\delta_{3s})^2},
\dfrac{2}{1\!+\!2\delta_{3s}\!+\!2\delta_{3s}^2}\!\right\}\!.\label{eq:obsp}
\end{eqnarray}
It is easy to calculate that when $\rho=\bar{\rho}=1/2$, Theorem \ref{thm:main_result} gives $\delta_{3s}=0.3063$ and $\tau=13.1303$, while (\ref{eq:obsp}) gives $\delta_{3s}=0.2324$ and $\bar{\tau}=21.1886$. Hence, the proposed result improves the theoretical guarantee for SP.

Now we give some discussions in another view. Substituting $n$ by $n-1$ in (\ref{eq:sp_6}), and then combining with (\ref{eq:sp_5}), we have
\begin{align}
&\Vert(\mathbf{x}_S)_{\overline{S^{n}}}\Vert_2\le\rho^{\prime}\Vert(\mathbf{x}_S)_{\overline{S^{n-1}}}\Vert_2+(1-\rho^{\prime})\tau^{\prime}\Vert\mathbf{e}^{\prime}\Vert_2. \label{eq:sp_metric2}
\end{align}
where
\begin{eqnarray}
\rho^{\prime}&=&\rho=\dfrac{\sqrt{2\delta_{3s}^2(1+\delta_{3s}^2)}}
{1-\delta_{3s}^2} < 1,\nonumber\\
(1-\rho^{\prime})\tau^{\prime}&=&\dfrac{\sqrt{2}\delta_{3s}\sqrt{1+\delta_{3s}}}
{1-\delta_{3s}}\left(\sqrt{\dfrac{1+\delta_{3s}^2}{1-\delta_{3s}}}+1\right)\nonumber\\
&&\quad +\; 2\sqrt{2(1+\delta_{3s})}+\dfrac{2\delta_{3s}}{\sqrt{1-\delta_{3s}}}.\label{eq:sp_metric2_arg}
\end{eqnarray}
Then it easily follows that
\begin{eqnarray}
\label{metric2}
\Vert(\mathbf{x}_S)_{\overline{S^{n}}}\Vert_2\le(\rho^{\prime})^n\Vert
\mathbf{x}_S\Vert_2+\tau^{\prime}\Vert\mathbf{e}^{\prime}\Vert_2.
\end{eqnarray}
From Remark \ref{rmk:orthogonality-rip}, we know there are two approximation metrics in the proofs of algorithm convergency. In our derivations, we can obtain recursion formulas of both approximation metrics, and both formulas have the same convergence rate, but different error coefficients. Note that (\ref{eq:sp_metric2}) and \cite[Theorem 10]{dai2009subspace} have the same forms. Rewrite \cite[Theorem 10]{dai2009subspace} as follows
\begin{eqnarray*}
\Vert(\mathbf{x}_S)_{\overline{S^{n}}}\Vert_2&\le&\bar{\rho}^{\,\prime}\Vert(\mathbf{x}_S)_{\overline{S^{n-1}}}\Vert_2+(1-\bar{\rho}^{\,\prime})\bar{\tau}^{\,\prime}\Vert\mathbf{e}^{\prime}\Vert_2,
\end{eqnarray*}
where
\begin{eqnarray}
\bar{\rho}^{\,\prime}=\dfrac{2\delta_{3s}(1+\delta_{3s})}{(1-\delta_{3s})^3},
\quad(1-\bar{\rho}^{\,\prime})\bar{\tau}^{\,\prime}
=\dfrac{4(1+\delta_{3s})}{(1-\delta_{3s})^2}.\label{eq:sp_original}
\end{eqnarray}
Then it easily follows that 
\begin{eqnarray*}
\Vert(\mathbf{x}_S)_{\overline{S^{n}}}\Vert_2&\le&
(\bar{\rho}^{\,\prime})^n\Vert\mathbf{x}_S\Vert_2
+\bar{\tau}^{\,\prime}\Vert\mathbf{e}^{\prime}\Vert_2.
\end{eqnarray*}
It is easy to see that when $\rho^{\prime}=\bar{\rho}^{\,\prime}=1/2$, our result (\ref{eq:sp_metric2_arg}) gives $\delta_{3s}=0.3063,\tau^{\prime}=11.3213$, while (\ref{eq:sp_original}) gives $\delta_{3s}=0.1397,\bar{\tau}^{\,\prime}=12.3219$. Hence, the proposed result improves the theoretical guarantee for SP.

\section{\label{sec:cosamp}Compressive Sampling Matching Pursuit}
To the best of our knowledge, Foucart \cite{foucart2012sparse} obtained the previously best-known results of RICs for CoSaMP. The main differences that help us get a further improved bound are that in the identification step of CoSaMP, we get a tighter bound than that of \cite{foucart2012sparse}, and by (\ref{eq:orthogonality-rip3}) of Lemma \ref{lem:orthogonality-rip}, we improve the bound in step 4 comparing with \cite{foucart2012sparse}.

Now, we turn to the well-known CoSaMP algorithm. CoSaMP is similar to SP, except that in the step 1, it adds $2s$ candidates to $\tilde{S}^{n}$, and in the final step, it omits the least squares procedure by directly keeping the $s$ largest magnitude entries of $\tilde{\mathbf{x}}^{n}$ and the corresponding support set. In order to compare easily with the previous section, we use the same symbols in both sections, but it is worthwhile to note that the symbols in both sections are completely independent.

Firstly, in the identification step, we have the following lemma for CoSaMP.
 \begin{lemma}[Identification for CoSaMP]In the steps 1 and 2 of CoSaMP,
\label{lem:identification-cosamp}
\[\Vert(\mathbf{x}_S)_{\overline{\tilde{S}^{n}}}\Vert_2\le\sqrt{2}\delta_{4s}\Vert\mathbf{x}_S-\mathbf{x}^{n-1}\Vert_2+\sqrt{2(1+\delta_{3s})}\Vert\mathbf{e}^{\prime}\Vert_2.\]
\end{lemma}

Lemma \ref{lem:identification-sp} and Lemma \ref{lem:identification-cosamp} have the similar forms, but their proofs are somewhat different, where the proof of Lemma \ref{lem:identification-sp} employs the property of orthogonality in Lemma \ref{lem:orthogonality}, but it's not necessary for Lemma \ref{lem:identification-cosamp}. The proof of Lemma \ref{lem:identification-cosamp} can be found in Appendix \ref{sec:proof-of-identification-cosamp}.

Then we change the derivations of SP slightly to get our main result for CoSaMP.

\begin{theorem}
\label{thm:cosamp}
For the general CS model $\mathbf{y}=\mathbf{\Phi}\mathbf{x}_S+\mathbf{e}^{\prime}$ in (\ref{eq:general_model}), if $\delta_{4s}<0.5$, then the sequence of $\mathbf{x}^{n}$ defined by CoSaMP satisfies
\begin{eqnarray}
\Vert\mathbf{x}_S-\mathbf{x}^{n}\Vert_2&\le&\rho^{n}\Vert\mathbf{x}_S\Vert_2
+\tau\Vert\mathbf{e}^{\prime}\Vert_2, \label{eq:cosamp_main_result}
\end{eqnarray}
where
\begin{eqnarray}
\rho&=&\sqrt{\dfrac{2\delta_{4s}^2(1+2\delta_{4s}^2)}{1-\delta_{4s}^2}}<1,
\label{eq:cosamp_tau0}\\
(1-\rho)\tau&=&\dfrac{(\sqrt{2}+1)\delta_{4s}(\sqrt{2(1-\delta_{4s})}
+\sqrt{1+\delta_{4s}})}{1-\delta_{4s}}\nonumber\\
&&\qquad+\;(2\sqrt{2}+1)\sqrt{1+\delta_{4s}}.  \label{eq:cosamp_tau}
\end{eqnarray}
\end{theorem}

\begin{proof}
The steps 1 and 2 of CoSaMP are the identification steps. By Remark \ref{rem22}, we assume that $|\tilde S^n|=3s$ and $\Delta S\cap S^{n-1} =\emptyset$ without loss of generality. By Lemma \ref{lem:identification-cosamp}, in the $n$-th iteration,
\begin{eqnarray}
\Vert(\mathbf{x}_S)_{\overline{\tilde{S}^{n}}}\Vert_2
\le\sqrt{2}\delta_{4s}\Vert\mathbf{x}_S-\mathbf{x}^{n-1}\Vert_2
\!+\!\sqrt{2(1+\delta_{3s})}\Vert\mathbf{e}^{\prime}\Vert_2.\label{eq:cosamp_1}
\end{eqnarray}
The step 3 of the $n$-th iteration is to solve a least squares problem. By (\ref{eq:orthogonality-rip2}) of Lemma \ref{lem:orthogonality-rip}, letting $T=\tilde{S}^{n}$ and $\mathbf{z}_p=\tilde{\mathbf{x}}^{n},t=3s$, we have
\begin{eqnarray}
\Vert\mathbf{x}_S-\tilde{\mathbf{x}}^{n}\Vert_2\le
\sqrt{\dfrac{1}{1-\delta_{4s}^2}}\Vert(\mathbf{x}_S)_{\overline{\tilde{S}^{n}}}
\Vert_2\!+\!\dfrac{\sqrt{1+\delta_{3s}}}{1-\delta_{4s}}\Vert\mathbf{e}^{\prime}\Vert_2.\label{eq:cosamp_2}
\end{eqnarray}
Combining (\ref{eq:cosamp_1}) and (\ref{eq:cosamp_2}), and magnifying $\delta_{3s}$ to $\delta_{4s}$, we have
\begin{eqnarray}
\!\!\!\!\Vert\mathbf{x}_S-\mathbf{\tilde{x}}^{n}\Vert_2\!\!&\le&\!\!\sqrt{\dfrac{2\delta_{4s}^2}
{1-\delta_{4s}^2}}\Vert\mathbf{x}_S-\mathbf{x}^{n-1}\Vert_2\nonumber\\
\!\!\!\!&&+\;\dfrac{\sqrt{2(1-\delta_{4s})}+\sqrt{1+\delta_{4s}}}{1-\delta_{4s}}\Vert\mathbf{e}^{\prime}\Vert_2.\label{eq:cosamp_3}
\end{eqnarray}
After the step 4 of the $n$-th iteration, define $S_{\nabla}:=\tilde{S}^{n}\backslash S^{n}$, where $S_{\nabla}$ contains the indices of the $2s$ smallest entries in $\mathbf{\tilde{x}}^{n}$. Letting $T=\tilde{S}^{n}$ and $\mathbf{z}_p=\tilde{\mathbf{x}}^{n},t=3s,T_{\nabla}=S_{\nabla}$, by (\ref{eq:orthogonality-rip3}) of Lemma \ref{lem:orthogonality-rip}, it follows that
\begin{eqnarray}
\Vert(\mathbf{x}_S)_{S_{\nabla}}\Vert_2\le \sqrt{2}\delta_{4s}\Vert\mathbf{x}_S-\tilde{\mathbf{x}}^{n}\Vert_2+\sqrt{2(1+\delta_{3s})}\Vert\mathbf{e}^{\prime}\Vert_2.\label{eq:cosamp_4}
\end{eqnarray}
Define $\tau_1=\dfrac{\sqrt{2(1-\delta_{4s})}+\sqrt{1+\delta_{4s}}} {1-\delta_{4s}},\tau_2=\sqrt{1+\delta_{4s}}$. \\
Dividing $\overline{S^{n}}$ into two disjoint parts: $S_{\nabla}, \overline{\tilde{S}^{n}}$, we have
\begin{eqnarray*}
\!\!\!\!\!\!\!\!&&\!\!\!\!\Vert(\mathbf{x}_S)_{\overline{S^{n}}}\Vert_2^2=\Vert(\mathbf{x}_S)_{S_{\nabla}}\Vert_2^2+\Vert(\mathbf{x}_S)_{\overline{\tilde{S}^{n}}}\Vert_2^2 \nonumber\\
\!\!\!\!\!\!\!\!&\overset{(\ref{eq:cosamp_4}),(\ref{eq:cosamp_1}) }{\le}&2\left(\delta_{4s}\Vert\mathbf{x}_{S}-\tilde{\mathbf{x}}^{n}\Vert_2
+\tau_{2}\Vert\mathbf{e}^{\prime}\Vert_2\right)^2\\
\!\!\!\!\!\!\!\!&&\quad+2\left(\delta_{4s}\Vert\mathbf{x}_{S}
-\mathbf{x}^{n-1}\Vert_2+\tau_{2}\Vert\mathbf{e}^{\prime}\Vert_2\right)^2\nonumber\\
\!\!\!\!\!\!\!\!&\overset{(\ref{eq:cosamp_3})}{\le}&\!\!\!\!\!\!\!\!\!\!
2\left(\!\delta_{4s}\sqrt{\dfrac{2\delta_{4s}^2}{1-\delta_{4s}^2}}
\Vert\mathbf{x}_S-\mathbf{x}^{n-1}\Vert_2+(\delta_{4s}\tau_1+\tau_2)
\Vert\mathbf{e}^{\prime}\Vert_2\!\right)^{\!\!2}\\
\!\!\!\!\!\!\!\!&&\quad+2\Bigg(\delta_{4s}\Vert\mathbf{x}_S
-\mathbf{x}^{n-1}\Vert_2+\tau_2\Vert\mathbf{e}^{\prime}\Vert_2\Bigg)^{\!\!2}\nonumber\\
\!\!\!\!\!\!\!\!&\overset{(\ref{lem3})}{\le}&\!\!\!\!\!\!\!\!\!\!
2\left(\!\!\sqrt{\dfrac{\delta_{4s}^2 (1+\delta_{4s}^2)}{1-\delta_{4s}^2}}\Vert\mathbf{x}_S
\!-\!\mathbf{x}^{n-1}\Vert_2+(\delta_{4s}\tau_1\!+\!2\tau_2)
\Vert\mathbf{e}^{\prime}\Vert_2\!\!\right)^{\!\!2}\!\!,
\end{eqnarray*}
which implies that
\begin{eqnarray}
\Vert(\mathbf{x}_{S})_{\overline{S^{n}}}\Vert_2&\le&
\sqrt{\dfrac{2\delta_{4s}^2 (1+\delta_{4s}^2)}{1-\delta_{4s}^2}} \Vert\mathbf{x}_S-\mathbf{x}^{n-1}\Vert_2\nonumber\\
&&\quad+\;\sqrt{2}(\delta_{4s}\tau_1+2\tau_2)\Vert\mathbf{e}^{\prime}\Vert_2.\label{eq:cosamp_5}
\end{eqnarray}
From the step 5 of the $n$-th iteration, we magnify $\Vert\mathbf{x}_S-\mathbf{x}^{n}\Vert_2$  in a different way from SP. Since $\mathbf{x}^{n}$ is obtained by keeping the $s$ largest magnitude entries of $\tilde{\mathbf{x}}^{n}$, we have
\begin{eqnarray}
\!\!\!\!\!\!\!\!&&\!\!\!\!\Vert(\mathbf{x}_S-\mathbf{x}^{n})_{S^{n}}\Vert_2\nonumber\\
\!\!\!\!\!\!\!\!&\le&\!\!\!\!\Vert(\mathbf{x}_S-\tilde{\mathbf{x}}^{n})_{\tilde{S}^{n}}\Vert_2\nonumber\\
\!\!\!\!\!\!\!\!&\overset{(\ref{eq:orthogonality-rip1})}{\le}&\!\!\!\!\delta_{4s}\Vert\mathbf{x}_S-\tilde{\mathbf{x}}^{n}\Vert_2+\sqrt{1+\delta_{4s}}\Vert\mathbf{e}^{\prime}\Vert_2\nonumber\\
\!\!\!\!\!\!\!\!&\overset{(\ref{eq:cosamp_3})}{\le}&\!\!\!\!\sqrt{\dfrac{2\delta_{4s}^4}
{1-\delta_{4s}^2}}\Vert\mathbf{x}_S-\mathbf{x}^{n-1}\Vert_2\nonumber\\
\!\!\!\!\!\!\!\!&&\!\!\!\!+\!\left(\!\delta_{4s}\dfrac{\sqrt{2(1-\delta_{4s})}
+\sqrt{1+\delta_{4s}}}{1-\delta_{4s}}+\sqrt{1+\delta_{4s}}\!\right)\!\Vert\mathbf{e}^{\prime}\Vert_2\nonumber\\
\!\!\!\!\!\!\!\!&=&\!\!\!\!\sqrt{\dfrac{2\delta_{4s}^4}{1-\delta_{4s}^2}}\Vert\mathbf{x}_S-\mathbf{x}^{n-1}\Vert_2+(\delta_{4s}\tau_1+\tau_2)\Vert\mathbf{e}^{\prime}\Vert_2.\label{eq:cosamp_6}
\end{eqnarray}
Dividing supp($\mathbf{x}_S-\mathbf{x}^{n}$) into two disjoint parts: $S^{n},\overline{S^{n}}$, and noticing that $(\mathbf{x}_S-\mathbf{x}^{n})_{\overline{S^{n}}}=(\mathbf{x}_S)_{\overline{S^{n}}}$, we have
\begin{eqnarray*}
\!\!\!\!\!\!\!\!&&\!\!\!\!\Vert\mathbf{x}_S-\mathbf{x}^{n}\Vert_2^2\\
\!\!\!\!\!\!&=&\!\!\!\!\Vert(\mathbf{x}_S-\mathbf{x}^{n})_{S^{n}}\Vert_2^2+\Vert(\mathbf{x}_S
-\mathbf{x}^{n})_{\overline{S^{n}}}\Vert_2^2\\
\!\!\!\!\!\!&=&\!\!\!\!\Vert(\mathbf{x}_S-\mathbf{x}^{n})_{S^{n}}\Vert_2^2+\Vert(\mathbf{x}_S)_{\overline{S^{n}}}\Vert_2^2\nonumber\\
\!\!\!\!\!\!&\overset{(\ref{eq:cosamp_6}),(\ref{eq:cosamp_5})}{\le}&
\!\!\!\!\!\!\Bigg(\sqrt{\dfrac{2\delta_{4s}^4}{1-\delta_{4s}^2}}
\Vert\mathbf{x}_S-\mathbf{x}^{n-1}\Vert_2+(\delta_{4s}\tau_1+\tau_2)
\Vert\mathbf{e}^{\prime}\Vert_2\Bigg)^{\!\!2}\\
\!\!\!\!\!\!\!\!\!\!&+&\!\!\!\!\!\!\!\!\!\!\!\!\left(\!\!
\sqrt{\!\dfrac{2\delta_{4s}^2 (1\!+\!\delta_{4s}^2)}{1\!-\!\delta_{4s}^2}}\Vert\mathbf{x}_S
\!-\!\mathbf{x}^{n-1}\Vert_2+\sqrt{2}(\delta_{4s}\tau_1 \!+\!2\tau_2)\Vert\mathbf{e}^{\prime}\Vert_2\!\!\right)^{\!\!2}\nonumber\\
\!\!\!\!\!\!\!\!\!\!&\overset{(\ref{lem3})}{\le}&\!\!\!\!\!\!\!\!
\Bigg(\sqrt{\dfrac{2\delta_{4s}^2(1+2\delta_{4s}^2)}{1-\delta_{4s}^2}}
\Vert\mathbf{x}_S-\mathbf{x}^{n-1}\Vert_2\\
\!\!\!\!\!\!\!\!\!\!&&+\;((\sqrt{2}+1)\delta_{4s}\tau_1+(2\sqrt{2}+1)\tau_2)\Vert\mathbf{e}^{\prime}\Vert_2\Bigg)^2,\label{eq:cosamp_7}
\end{eqnarray*}
or
\begin{eqnarray*}
\Vert\mathbf{x}_S-\mathbf{x}^{n}\Vert_2
&\le&\sqrt{\dfrac{2\delta_{4s}^2(1+2\delta_{4s}^2)}{1-\delta_{4s}^2}}
\Vert\mathbf{x}_S-\mathbf{x}^{n-1}\Vert_2\\
&&+\;((\sqrt{2}+1)\delta_{4s}\tau_1+(2\sqrt{2}+1)\tau_2)\Vert\mathbf{e}^{\prime}\Vert_2\\
&=&\rho\Vert\mathbf{x}_S-\mathbf{x}^{n-1}\Vert_2+(1-\rho)\tau\Vert\mathbf{e}^{\prime}\Vert_2,
\end{eqnarray*}
where $\rho$ and $\tau$ is respectively referred to (\ref{eq:cosamp_tau0}) and (\ref{eq:cosamp_tau}). Hence, (\ref{eq:cosamp_main_result}) follows by recursively using the above inequality when $\rho<1$. Note that $\rho<1$ if \begin{equation*}
4\delta_{4s}^4+3\delta_{4s}^2-1<0\qquad \mbox{or} \quad \delta_{4s}<1/2=0.5,
\end{equation*}
which finishes the proof.
\end{proof}

\begin{remark}
While Foucart \cite{foucart2012sparse} gives $\delta_{4s}<0.38427 $ with $\rho<1$ and $\delta_{4s}<0.22665$ with $\rho<1/2$, it is easy to see from Theorem \ref{thm:cosamp} that $\delta_{4s}<0.5$ with $\rho<1$ and $\delta_{4s}<0.3083$ with $\rho<1/2$. Hence, the proposed result improves the theoretical guarantee for CoSaMP.
\end{remark}

\section{\label{sec:conclusion}Conclusion}
In this paper, we improve the RICs for both the SP and CoSaMP algorithms. Firstly, for the $s$-sparse recovery, the RICs for SP are enlarged to $\delta_{3s}<0.4859$ with convergence rate $\rho<1$ and $\delta_{3s}<0.3063$  with $\rho<1/2$. Moreover, we show that the recursive formula (\ref{thm-eq}) by the approximation metric $\Vert\mathbf{x}_S-\mathbf{x}^{n}\Vert_2$ and the recursive formula (\ref{metric2}) by the approximation metric $\Vert(\mathbf{x}_S)_{\overline{S^{n}}}\Vert_2$ have the same convergence rate $\rho$.
Then, we deal with the CoSaMP algorithm and show that for the $s$-sparse recovery, the RICs for CoSaMP can be enlarged to $\delta_{4s}<0.5$ with $\rho<1$ and $\delta_{4s}<0.3083$ with $\rho<1/2$. Very recently, \cite{cai2012sharp,cai2013sparse} get sharp RIP bounds for BP. One may wonder whether similar results could be obtained for greedy pursuits or not. Future works may focus on the sharp RIP bounds for greedy pursuits.

\appendix

\subsection{\label{sec:proof-of-rip}Proof of Lemma \ref{lem:rip}}

\begin{enumerate}
\item
By the definition of RIC and the fact that an $s$-sparse vector is also an $s^{\prime}$-sparse vector, we have for any $s$-sparse vector $\mathbf{x}$,
\[
(1-\delta_{s^{\prime}})\Vert\mathbf{x}\Vert_2^2\le\Vert\mathbf{\Phi}\mathbf{x}\Vert_2^2\le(1+\delta_{s^{\prime}})\Vert\mathbf{x}\Vert_2^2 .\]
Since $\delta_s$ is the infimum of all parameters satisfying (\ref{eq:origin_def}), $$\delta_s\le\delta_{s^{\prime}}.$$
\item Let $T=\text{supp}(\mathbf{u})\cup\text{supp}(\mathbf{v})$. Then $|T|\le t$. Let $\mathbf{u}_{|T},\mathbf{v}_{|T}$ denote respectively the $T$-dimensional sub-vectors of $\mathbf{u}$ and $\mathbf{v}$ obtained by only keeping the components indexed by $T$. It follows that
    \begin{eqnarray}
    &&|\langle\mathbf{u},(\mathbf{I}-\mathbf{\Phi}^{\! *}\mathbf{\Phi})\mathbf{v}\rangle| \nonumber\\ &=&|\langle\mathbf{u},\mathbf{v}\rangle-\langle\mathbf{\Phi}
    \mathbf{u},\mathbf{\Phi}\mathbf{v}\rangle|\nonumber\\
    &=&|\langle\mathbf{u}_{|T},\mathbf{v}_{|T}\rangle-\langle\mathbf{\Phi}_{T}
    \mathbf{u}_{|T},\mathbf{\Phi}_{T}\mathbf{v}_{|T}\rangle|\nonumber\\
    &=&|\langle\mathbf{u}_{|T},\mathbf{v}_{|T}\rangle-\langle\mathbf{u}_{|T},
    \mathbf{\Phi}_{T}^{\! *}\mathbf{\Phi}_{T}\mathbf{v}_{|T}\rangle|\nonumber\\
    &=&|\langle\mathbf{u}_{|T},(\mathbf{I}-\mathbf{\Phi}^{\! *}_{T}\mathbf{\Phi}_{T})\mathbf{v}_{|T}\rangle|\nonumber\\
    &\le& \Vert\mathbf{u}_{|T}\Vert_2\Vert(\mathbf{I}-\mathbf{\Phi}_{T}^{\! *}\mathbf{\Phi}_{T})\mathbf{v}_{|T}\Vert_2\label{p1}\\
    &\overset{(\ref{eq:new_def2})}{\le}&\Vert\mathbf{u}_{|T}\Vert_2\Vert\mathbf{I}-\mathbf{\Phi}_{T}^{\! *}\mathbf{\Phi}_{T}\Vert_{2\rightarrow 2}\Vert\mathbf{v}_{|T}\Vert_2\nonumber\\
    &\overset{(\ref{eq:new_def})}{\le}&\delta_{t}\Vert\mathbf{u}_{|T}
    \Vert_2\Vert\mathbf{v}_{|T}\Vert_2\nonumber\\
    &=&\delta_{t}\Vert\mathbf{u}\Vert_2\Vert\mathbf{v}\Vert_2,\nonumber
    \end{eqnarray}
where $(\ref{p1})$ is from the Cauchy-Schwartz inequality, and the inequality (\ref{rip11}) follows.
Moreover,
    \begin{eqnarray*}
    &&\Vert((\mathbf{I}-\mathbf{\Phi}^{\! *}\mathbf{\Phi})\mathbf{v})_{U}\Vert_2^2\\
    &=&\langle((\mathbf{I}-\mathbf{\Phi}^{\! *}\mathbf{\Phi})\mathbf{v})_{U},(\mathbf{I}-\mathbf{\Phi}^{\! *}\mathbf{\Phi})\mathbf{v}\rangle\\
    &\overset{(\ref{rip11})}{\le}&\delta_{t}\Vert((\mathbf{I}-\mathbf{\Phi}^{\! *}\mathbf{\Phi})\mathbf{v})_{U}\Vert_2\Vert\mathbf{v}\Vert_2,
    \end{eqnarray*}
    which implies the inequality (\ref{rip12}).
\end{enumerate}

\subsection{\label{sec:proof-of-orghononality-rip}Proof of Lemma \ref{lem:orthogonality-rip}}
\begin{enumerate}
\item
By Remark \ref{rem2} of Lemma \ref{lem:orthogonality}, letting $$\mathbf{z}=(\mathbf{x}_S-\mathbf{z}_{p})_{T},$$ we have
\begin{eqnarray}
&&\langle \mathbf{x}_S-\mathbf{z}_{p}, \mathbf{\Phi}^{ *}\mathbf{\Phi}(\mathbf{x}_S-\mathbf{z}_{p})_{T}\rangle\nonumber\\
&&\qquad\qquad\quad+\langle\mathbf{e}^{\prime},\mathbf{\Phi}(\mathbf{x}_S-\mathbf{z}_{p})_{T}\rangle=0. \label{eq:tmp1}
\end{eqnarray}
Noticing that $$\text{supp}(\mathbf{z}_{p})\subseteq T, \text{supp}(\mathbf{x}_S-\mathbf{z}_{p})\subseteq S\cup T, \text{supp}((\mathbf{x}_S-\mathbf{z}_{p})_T)\subseteq T,$$ we have
\begin{eqnarray}
&&\Vert(\mathbf{x}_S-\mathbf{z}_{p})_{T}\Vert_2^2\nonumber\\
&=&\langle \mathbf{x}_S-\mathbf{z}_{p}, (\mathbf{x}_S-\mathbf{z}_{p})_{T} \rangle \nonumber\\
&\overset{(\ref{eq:tmp1})}{=}&\langle\mathbf{x}_S
-\mathbf{z}_{p},(\mathbf{I}-\mathbf{\Phi}^{ *}\mathbf{\Phi})(\mathbf{x}_S-\mathbf{z}_{p})_{T}\rangle\nonumber\\
&&\qquad-\langle\mathbf{e^{\prime}},\mathbf{\Phi}(\mathbf{x}_S
-\mathbf{z}_{p})_{T}\rangle\nonumber\\
&\overset{(\ref{rip11})}{\le}&\delta_{s+t}\Vert(\mathbf{x}_S
-\mathbf{z}_{p})_{T}\Vert_2\Vert\mathbf{x}_S-\mathbf{z}_{p}\Vert_2\nonumber\\
&&\qquad +|\langle\mathbf{e^{\prime}},\mathbf{\Phi}(\mathbf{x}_S
-\mathbf{z}_{p})_{T}\rangle|\nonumber\\
&\le&\delta_{s+t}\Vert(\mathbf{x}_S-\mathbf{z}_{p})_{T}
\Vert_2\Vert\mathbf{x}_S-\mathbf{z}_{p}\Vert_2\nonumber\\
&&\qquad + \Vert\mathbf{e}^{\prime}\Vert_2\Vert\mathbf{\Phi}
(\mathbf{x}_S-\mathbf{z}_{p})_{T}\Vert_2\label{p2}\\
&\overset{(\ref{eq:origin_def})}{\le}&\delta_{s+t}\Vert(\mathbf{x}_S
-\mathbf{z}_{p})_{T}\Vert_2\Vert\mathbf{x}_S-\mathbf{z}_{p}\Vert_2\nonumber\\
&&\qquad + \Vert\mathbf{e}^{\prime}\Vert_2\sqrt{1+\delta_{t}} \Vert(\mathbf{x}_S-\mathbf{z}_{p})_{T}\Vert_2,\label{t11}
\end{eqnarray}
where the inequality (\ref{p2}) is from the well-known Cauchy-Schwartz inequality.
After both sides of (\ref{t11}) are divided by $\Vert(\mathbf{x}_S-\mathbf{z}_{p})_{T}\Vert_2$, the claim (\ref{eq:orthogonality-rip1}) in the lemma follows.

\item
By dividing the indices of $\mathbf{x}_S-\mathbf{z}_{p}$ into two disjoint parts: $T$ and $\overline{T}$, we find the relations between $\Vert\mathbf{x}_S-\mathbf{z}_{p}\Vert_2$ and $\Vert(\mathbf{x}_S)_{\overline{T}}\Vert_2$ as follows.
Noticing that $$\text{supp}(\mathbf{z}_p)\subseteq T \mbox{ and } \Vert(\mathbf{x}_S-\mathbf{z}_{p})_{\overline{T}}\Vert_2
=\Vert(\mathbf{x}_S)_{\overline{T}}\Vert_2,$$
we have
\begin{align}
&\quad\;\Vert\mathbf{x}_S-\mathbf{z}_{p}\Vert_2^2\nonumber\\
&=\Vert(\mathbf{x}_S-\mathbf{z}_{p})_{\overline{T}}\Vert_2^2+\Vert(\mathbf{x}_S-\mathbf{z}_{p})_{T}\Vert_2^2\nonumber \\
&=\Vert(\mathbf{x}_S)_{\overline{T}}\Vert_2^2+\Vert(\mathbf{x}_S-\mathbf{z}_{p})_{T}\Vert_2^2\nonumber \\
&\overset{(\ref{eq:orthogonality-rip1})}{\le}\Vert(\mathbf{x}_S)_{\overline{T}}\Vert_2^2+(\delta_{s+t}\Vert\mathbf{x}_S-\mathbf{z}_{p}\Vert_2+\sqrt{1+\delta_{t}}\Vert\mathbf{e}^{\prime}\Vert_2)^2.\qquad\nonumber
\end{align}
Define $\omega:=\Vert\mathbf{x}_S-\mathbf{z}_{p}\Vert_2$. After arrangement, we have
\begin{align}
&(1-\delta_{s+t}^2)\omega^2-2\delta_{s+t}\sqrt{1+\delta_{t}}
\Vert\mathbf{e}^{\prime}\Vert_2 \omega \nonumber\\ &\qquad\qquad-((1+\delta_{t})\Vert\mathbf{e}^{\prime}\Vert_2^2
+\Vert(\mathbf{x}_S)_{\overline{T}}\Vert_2^2)\;\le\;0.\label{eq:22}
\end{align}
Solving the quadratic inequality (\ref{eq:22}) with $\omega$, we have
\begin{align}
&\Vert\mathbf{x}_S-\mathbf{z}_{p}\Vert_2 =\omega\le\dfrac{\delta_{s+t}\sqrt{1+\delta_{t}}
\Vert\mathbf{e}^{\prime}\Vert_2}{1-\delta_{s+t}^2} \nonumber\\ &\qquad\quad+\dfrac{\sqrt{(1+\delta_{t})\Vert\mathbf{e}^{\prime}\Vert_2^2+(1-\delta_{s+t}^2)\Vert(\mathbf{x}_S)_{\overline{T}}\Vert_2^2}}{1-\delta_{s+t}^2}.\nonumber
\end{align}
Using the inequality $\sqrt{a^2+b^2}\le a+b$ for $a,b\ge 0$ and after a little simplification, we have
\begin{align}
\Vert\mathbf{x}_S-\mathbf{z}_{p}\Vert_2\le\sqrt{\dfrac{1}{1-\delta_{s+t}^2}}\Vert(\mathbf{x}_S)_{\overline{T}}\Vert_2+\dfrac{\sqrt{1+\delta_{t}}}{1-\delta_{s+t}}\Vert\mathbf{e}^{\prime}\Vert_2.\nonumber
\end{align}
This completes the proof of (\ref{eq:orthogonality-rip2}) in the lemma.
\item
The basic idea is to find a subset $T^{\prime}\subseteq T$ such that $T^{\prime}\cap S=\emptyset$. This idea is initially proposed in Dai and Milenkovic \cite{dai2009subspace}.
However, we get a tighter upper bound of $\Vert(\mathbf{x}_S)_{T_{\nabla}}\Vert$.

Since $t>s$, there is a set $T^{\prime}\subseteq T\setminus S$ with  $|T^{\prime}|=t-s$. Since $T_{\nabla}$ is defined by the set of indices of the $t-s$ smallest entries of $\mathbf{z}_{p}$ in $T$, we have
\begin{align*}
\Vert(\mathbf{z}_{p})_{T_{\nabla}}\Vert_2 \le\Vert(\mathbf{z}_{p})_{T^{\prime} }\Vert_2.
\end{align*}
By eliminating the contribution on $T_{\nabla}\cap T^{\prime}$, we have
\begin{eqnarray}
\Vert(\mathbf{z}_{p})_{T_{\nabla}\!\backslash \!T^{\prime}}\!\Vert_2 \le \Vert(\mathbf{z}_{p})_{T^{\prime}\!\backslash\! T_{\nabla}}\!\Vert_2
=\Vert(\mathbf{z}_{p}\!-\!\mathbf{x}_S)_{T^{\prime}\!\backslash \! T_{\nabla}}\!\Vert_2,\label{eq:3}
\end{eqnarray}
where the last equality is from $$S\cap T'=\emptyset \mbox{ and } (\mathbf{x}_S)_{T^{\prime}\backslash T_{\nabla}}=\mathbf{0}.$$
For the left-hand side of (\ref{eq:3}), noticing that $$S\cap T'=\emptyset \mbox{ and } (\mathbf{x}_S)_{T_{\nabla}\backslash T^{\prime}}=(\mathbf{x}_S)_{T_{\nabla}},$$ we have
\begin{align}
&\quad\;\Vert(\mathbf{z}_{p})_{T_{\nabla}\backslash T^{\prime}}\Vert_2\nonumber\\
& =\Vert(\mathbf{z}_{p}-\mathbf{x}_S)_{T_{\nabla}\backslash T^{\prime}}+(\mathbf{x}_S)_{T_{\nabla}\backslash T^{\prime}}\Vert_2  \nonumber\\
&=\Vert(\mathbf{z}_{p}-\mathbf{x}_S)_{T_{\nabla}\backslash T^{\prime}}+(\mathbf{x}_S)_{T_{\nabla}}\Vert_2\nonumber\\
&\ge\Vert(\mathbf{x}_S)_{T_{\nabla}}\Vert_2-\Vert(\mathbf{z}_{p}-\mathbf{x}_S)_{T_{\nabla}\backslash T^{\prime}}\Vert_2. \label{eq:32}
\end{align}
Combining (\ref{eq:3}) and (\ref{eq:32}), and noticing $$(T_{\nabla}\backslash T^{\prime} )\cap(T^{\prime} \backslash T_{\nabla})=\emptyset \mbox{ and } (T_{\nabla}\backslash T^{\prime} )\cup(T^{\prime} \backslash T_{\nabla})\subseteq T,$$ we have
\begin{align}
&\quad\;\Vert(\mathbf{x}_S)_{T_{\nabla}}\Vert_2 \nonumber\\ &\le\Vert(\mathbf{z}_{p}-\mathbf{x}_S)_{T_{\nabla}\backslash T^{\prime}}\Vert_2+\Vert(\mathbf{z}_{p}-\mathbf{x}_S)_{T^{\prime} \backslash T_{\nabla}}\Vert_2 \nonumber\\
&\le \sqrt{2}\Vert(\mathbf{z}_{p}-\mathbf{x}_S)_{(T_{\nabla}\backslash T^{\prime} )\cup{(T^{\prime} \backslash T_{\nabla})}}\Vert_2\label{p3}\\
&\le \sqrt{2}\Vert(\mathbf{z}_{p}-\mathbf{x}_S)_{T}\Vert_2\nonumber\\
&\overset{(\ref{eq:orthogonality-rip1})}{\le} \sqrt{2}\delta_{s+t}\Vert\mathbf{x}_S-\mathbf{z}_{p}\Vert_2+\sqrt{2(1+\delta_{t})}\Vert\mathbf{e}^{\prime}\Vert_2,\nonumber
\end{align}
where the inequality $(\ref{p3})$ is from the Cauchy-Schwartz inequality, and the claim (\ref{eq:orthogonality-rip3}) in the lemma follows.
\end{enumerate}

\subsection{\label{sec:proof-of-identification-sp}Proof of Lemma \ref{lem:identification-sp}}
From the step 5 of the $(n-1)$-th iteration, $$\mathbf{x}^{n-1}=\mbox{arg}\min_{\mathbf{z}\in\mathbb{R}^{N}}
\{\Vert\mathbf{y}-\mathbf{\Phi}\mathbf{z}\Vert_2,
\text{supp}(\mathbf{z})\subseteq S^{n-1}\}.$$
By Lemma \ref{lem:orthogonality},
\begin{equation}
(\mathbf{\Phi}^{\! *}(\mathbf{y}-\mathbf{\Phi}\mathbf{x}^{n-1}))_{S^{n-1}}=\mathbf{0}.\label{eq:basic_observation}
\end{equation}

From the step 1 of the $n$-th iteration, $\Delta S$ is the set of $s$ indices corresponding to the $s$ largest magnitude entries in $ \mathbf{\Phi}^{\! *}\, (\mathbf{y}- \mathbf{\Phi} \mathbf{x}^{n-1}).$
Thus,
\begin{align}
\Vert(\mathbf{\Phi}^{\! *}(\mathbf{y}-\mathbf{\Phi}\mathbf{x}^{n-1}))_S\Vert_2 \le\Vert(\mathbf{\Phi}^{\! *}(\mathbf{y}-\mathbf{\Phi}\mathbf{x}^{n-1})_{\Delta S}\Vert_2.\nonumber
\end{align}
Removing the common coordinates in $S \cap \Delta S$, we have
\begin{eqnarray}
&\!\!\!\!\!\!\!\!\Vert(\mathbf{\Phi}^{\! *}\!(\mathbf{y}\!-\!\mathbf{\Phi}\mathbf{x}^{n-1}))_{S\backslash\Delta S}\Vert_2 \le\Vert(\mathbf{\Phi}^{\! *}\!(\mathbf{y}\!-\!\mathbf{\Phi}\mathbf{x}^{n-1}))_{\Delta S\backslash S}\Vert_2.\label{eq:1}
\end{eqnarray}
Since $\text{supp}(\mathbf{x}_S)\subseteq S$ and $\text{supp}(\mathbf{x}^{n-1})\subseteq S^{n-1}$,
\begin{eqnarray}
(\mathbf{x}_S-\mathbf{x}^{n-1})_{\Delta S\backslash(S\cup S^{n-1})}=\mathbf{0}.\label{lem6-1}
\end{eqnarray}
For the right-hand side of (\ref{eq:1}), we have
\begin{eqnarray}
\!\!\!\!\!\!\!\!\!\!\!\!\!\!\!\!&&\!\!\!\!\Vert(\mathbf{\Phi}^{\! *}(\mathbf{y}-\mathbf{\Phi} \mathbf{x}^{n-1}))_{\Delta S\backslash S}\Vert_2 \nonumber\\
\!\!\!\!\!\!\!\!\!\!\!\!\!\!\!\!&\overset{(\ref{eq:basic_observation})}{=}&\!\!\!\!\Vert(\mathbf{\Phi}^{\! *}(\mathbf{y}-\mathbf{\Phi}\mathbf{x}^{n-1}))_{\Delta S\backslash(S\cup S^{n-1})}\Vert_2 \nonumber\\
\!\!\!\!\!\!\!\!\!\!\!\!\!\!\!\!&=&\!\!\!\!\Vert(\mathbf{\Phi}^{\! *}(\mathbf{\Phi}\mathbf{x}_S+\mathbf{e}^{\prime}-\mathbf{\Phi}\mathbf{x}^{n-1}))_{\Delta S\backslash(S\cup S^{n-1})}\Vert_2\nonumber\\
\!\!\!\!\!\!\!\!\!\!\!\!\!\!\!\!&\overset{(\ref{lem6-1})}{=}&\!\!\!\!\Vert((\mathbf{\Phi}^{\! *}\mathbf{\Phi}-\mathbf{I})(\mathbf{x}_S-\mathbf{x}^{n-1})+\mathbf{\Phi}^{\! *}\mathbf{e}^{\prime})_{\Delta S\backslash(S\cup S^{n-1})}\Vert_2\nonumber\\
\!\!\!\!\!\!\!\!\!\!\!\!\!\!\!\!&\le&\!\!\!\!\Vert((\mathbf{\Phi}^{\! *}\mathbf{\Phi}-\mathbf{I})(\mathbf{x}_S-\mathbf{x}^{n-1})+\mathbf{\Phi}^{\! *}\mathbf{e}^{\prime})_{\Delta S\backslash S}\Vert_2\nonumber\\
\!\!\!\!\!\!\!\!\!\!\!\!\!\!\!\!&\le&\!\!\!\!\!\!\Vert((\mathbf{\Phi}^{\! *}\mathbf{\Phi}\!-\!\mathbf{I})(\mathbf{x}_S\!-\!\mathbf{x}^{n-1}))_{\Delta S\backslash S}\Vert_2\!+\Vert(\mathbf{\Phi}^{\! *}\mathbf{e}^{\prime})_{\Delta S\backslash S}\Vert_2.\label{eq:11}
\end{eqnarray}
From the step 2 of the $n$-th iteration, $\tilde{S}^{n}=S^{n-1}\cup\Delta S$. \\
Since $\text{supp}(\mathbf{x}^{n-1})\subseteq S^{n-1}\subseteq \tilde{S^n}$,
\begin{eqnarray}
(\mathbf{x}_S-\mathbf{x}^{n-1})_{S\backslash \tilde{S^n}}=(\mathbf{x}_S)_{\overline{\tilde{S^n}}}.\label{lem6-2}
\end{eqnarray}
For the left-hand side of (\ref{eq:1}), we have
\begin{eqnarray}
&&\Vert(\mathbf{\Phi}^{\! *}(\mathbf{y}-\mathbf{\Phi}\mathbf{x}^{n-1}))_{S\backslash\Delta S}\Vert_2\nonumber\\
&\overset{(\ref{eq:basic_observation})}{=}&\Vert(\mathbf{\Phi}^{\! *}(\mathbf{y}-\mathbf{\Phi}\mathbf{x}^{n-1}))_{S\backslash(\Delta S\cup S^{n-1})}\Vert_2\nonumber\\
&=&\Vert(\mathbf{\Phi}^{\! *}(\mathbf{\Phi}\mathbf{x}_S+\mathbf{e}^{\prime}-\mathbf{\Phi}\mathbf{x}^{n-1}))_{S\backslash\tilde{S}^{n}}\Vert_2\nonumber\\
&\overset{(\ref{lem6-2})}{=}&\Vert((\mathbf{\Phi}^{\! *}\mathbf{\Phi}-\mathbf{I})(\mathbf{x}_S-\mathbf{x}^{n-1}))_{S\backslash \tilde{S}^{n}}\nonumber\\
&&\quad+\;(\mathbf{x}_S)_{\overline{\tilde{S}^{n}}}+(\mathbf{\Phi}^{\! *}\mathbf{e}^{\prime})_{S\backslash \tilde{S}^{n}} \Vert_2\nonumber\\
&\ge&\Vert(\mathbf{x}_S)_{\overline{\tilde{S}^{n}}}\Vert_2
-\Vert(\mathbf{\Phi}^{\! *}\mathbf{e}^{\prime})_{S\backslash\tilde{S}^{n}}\Vert_2\nonumber\\
&&\quad-\;\Vert((\mathbf{\Phi}^{\! *}\mathbf{\Phi}-\mathbf{I}) (\mathbf{x}_S-\mathbf{x}^{n-1}))_{S\backslash\tilde{S}^{n}}\Vert_2
.\label{eq:12}
\end{eqnarray}
Combining (\ref{eq:1}), (\ref{eq:11}) and (\ref{eq:12}), and noticing that $$(\Delta S\backslash S)\cap(S\backslash\tilde{S}^{n})=\emptyset,$$
we have
\begin{eqnarray}
&&\!\!\!\!\Vert(\mathbf{x}_S)_{\overline{\tilde{S}^{n}}}\Vert_2\nonumber\\
&\le&\!\!\!\!\Vert((\mathbf{\Phi}^{\! *}\mathbf{\Phi}-\mathbf{I})(\mathbf{x}_S-\mathbf{x}^{n-1}))_{\Delta S\backslash S}\Vert_2\nonumber\\
&&\!\!\!\!\qquad+\Vert((\mathbf{\Phi}^{\! *}\mathbf{\Phi}-\mathbf{I})(\mathbf{x}_S
-\mathbf{x}^{n-1}))_{S\backslash\tilde{S}^{n}}\Vert_2\nonumber\\
&&\!\!\!\!\qquad+\Vert(\mathbf{\Phi}^{\! *}\mathbf{e}^{\prime})_{\Delta S\backslash S}\Vert_2+\Vert(\mathbf{\Phi}^{\! *}\mathbf{e}^{\prime})_{S\backslash \tilde{S}^{n}}\Vert_2\nonumber\\
&\le&\!\!\!\!\sqrt{2}\Vert((\mathbf{\Phi}^{\! *}\mathbf{\Phi}-\mathbf{I})(\mathbf{x}_S-\mathbf{x}^{n-1}))_{(\Delta S\backslash S)\cup(S\backslash\tilde{S}^{n})}\Vert_2\nonumber\\
&&\!\!\!\!\qquad+\sqrt{2}\Vert(\mathbf{\Phi}^{\! *}\mathbf{e}^{\prime})_{(\Delta S\backslash S)\cup(S\backslash \tilde{S}^{n}) }\Vert_2\label{p4}\\
&\le&\!\!\!\!\sqrt{2}\Vert((\mathbf{\Phi}^{\! *}\mathbf{\Phi}-\mathbf{I})(\mathbf{x}_S-\mathbf{x}^{n-1}))_{\Delta S\cup S}\Vert_2\nonumber\\
&&\!\!\!\!\qquad+\sqrt{2}\Vert(\mathbf{\Phi}^{\! *}\mathbf{e}^{\prime})_{\Delta S\cup S }\Vert_2\nonumber\\
&\overset{(\ref{rip12}), (\ref{rip13})}{\le}&\!\!\!\!\sqrt{2}\delta_{3s}\Vert\mathbf{x}_S- \mathbf{x}^{n-1}\Vert_2+\sqrt{2(1+\delta_{2s})}\Vert\mathbf{e}^{\prime}\Vert_2,\nonumber
\end{eqnarray}
where the inequality $(\ref{p4})$ is from the Cauchy-Schwartz inequality. This completes the proof of Lemma \ref{lem:identification-sp}.

\subsection{\label{sec:proof-of-identification-cosamp}Proof of Lemma \ref{lem:identification-cosamp}}
From the step 1 of the $n$-th iteration, $\Delta S$ is the set of the $2s$ indices corresponding to the  $2s$ largest magnitude entries in $ \mathbf{\Phi}^{\! *}\, (\mathbf{y}- \mathbf{\Phi} \mathbf{x}^{n-1})$.
Thus,
\begin{eqnarray}
\Vert(\mathbf{\Phi}^{\! *}(\mathbf{y}-\mathbf{\Phi}\mathbf{x}^{n-1}))_{S\cup S^{n-1}}\Vert_2\le\Vert(\mathbf{\Phi}^{\! *} (\mathbf{y}-\mathbf{\Phi}\mathbf{x}^{n-1}))_{\Delta S}\Vert_2.\nonumber
\end{eqnarray}
Removing the common coordinates in $(S\cup S^{n-1}) \cap \Delta S$ and noticing that $\mathbf{y}=\mathbf{\Phi}\mathbf{x}_S+\mathbf{e}^{\prime}$, we have
\begin{eqnarray}
&&\!\!\!\!\!\!\!\!\!\!\!\!\!\!\!\!\!\!\!\!\!\!\!\!\Vert(\mathbf{\Phi}^{\! *}(\mathbf{\Phi}\mathbf{x}_S+\mathbf{e}^{\prime}-\mathbf{\Phi}\mathbf{x}^{n-1}))_{(S\cup S^{n-1})\backslash\Delta S}\Vert_2\nonumber\\
&\le&\Vert(\mathbf{\Phi}^{\! *}(\mathbf{\Phi}\mathbf{x}_S+\mathbf{e}^{\prime}-\mathbf{\Phi}\mathbf{x}^{n-1}))_{\Delta S\backslash (S\cup S^{n-1})}\Vert_2.\label{eq:cosamp_first_1}
\end{eqnarray}
For the right-hand side of (\ref{eq:cosamp_first_1}), noticing that $$(\mathbf{x}_S-\mathbf{x}^{n-1})_{\Delta S\backslash(S\cup S^{n-1})}=\mathbf{0},$$
we have
\begin{eqnarray}
&&\Vert(\mathbf{\Phi}^{\! *}(\mathbf{\Phi}\mathbf{x}_S+\mathbf{e}^{\prime}-\mathbf{\Phi}\mathbf{x}^{n-1}))_{\Delta S\backslash(S\cup S^{n-1})}\Vert_2\nonumber\\
&=&\Vert((\mathbf{\Phi}^{\! *}\mathbf{\Phi}-\mathbf{I})(\mathbf{x}_S-\mathbf{x}^{n-1}))_{\Delta S\backslash(S\cup S^{n-1})}\nonumber\\
&&\qquad+(\mathbf{\Phi}^{\! *}\mathbf{e}^{\prime})_{\Delta S\backslash(S\cup S^{n-1})}\Vert_2\nonumber\\
&\le&\Vert((\mathbf{\Phi}^{\! *}\mathbf{\Phi}-\mathbf{I})(\mathbf{x}_S-\mathbf{x}^{n-1}))_{\Delta S\backslash (S\cup S^{n-1})}\Vert_2\nonumber\\
&&\qquad+\Vert(\mathbf{\Phi}^{\! *}\mathbf{e}^{\prime})_{\Delta S\backslash (S\cup S^{n-1})}\Vert_2.\label{eq:cosamp_first_11}
\end{eqnarray}
For the left-hand side of (\ref{eq:cosamp_first_1}), noticing that $$(\mathbf{x}_S-\mathbf{x}^{n-1})_{(S\cup S^{n-1})\backslash\Delta S}=(\mathbf{x}_S-\mathbf{x}^{n-1})_{\overline{\Delta S}},$$
we have
\begin{align}
&\quad\;\Vert(\mathbf{\Phi}^{\! *}(\mathbf{\Phi}\mathbf{x}_S+\mathbf{e}^{\prime}-\mathbf{\Phi}\mathbf{x}^{n-1}))_{(S\cup S^{n-1})\backslash\Delta S}\Vert_2\nonumber\\
&=\Vert(\mathbf{\Phi}^{\! *}\mathbf{\Phi}(\mathbf{x}_S-\mathbf{x}^{n-1}))_{(S\cup S^{n-1})\backslash \Delta S}\nonumber\\
&\qquad+(\mathbf{\Phi}^{\! *}\mathbf{e}^{\prime})_{(S\cup S^{n-1})\backslash \Delta S}\Vert_2\nonumber\\
&=\Vert((\mathbf{\Phi}^{\! *}\mathbf{\Phi}-\mathbf{I})(\mathbf{x}_S-\mathbf{x}^{n-1}))_{(S\cup S^{n-1})\backslash \Delta S}\nonumber\\
&\qquad+(\mathbf{x}_S-\mathbf{x}^{n-1})_{(S\cup S^{n-1})\backslash\Delta S}+(\mathbf{\Phi}^{\! *}\mathbf{e}^{\prime})_{(S\cup S^{n-1})\backslash\Delta S} \Vert_2\nonumber\\
&=\Vert((\mathbf{\Phi}^{\! *}\mathbf{\Phi}-\mathbf{I})(\mathbf{x}_S-\mathbf{x}^{n-1}))_{(S\cup S^{n-1})\backslash \Delta S}\nonumber\\
&\qquad+(\mathbf{x}_S-\mathbf{x}^{n-1})_{\overline{\Delta S}}+(\mathbf{\Phi}^{\! *}\mathbf{e}^{\prime})_{(S\cup S^{n-1})\backslash\Delta S} \Vert_2\nonumber\\
&\ge\Vert(\mathbf{x}_S-\mathbf{x}^{n-1})_{\overline{\Delta S}}\Vert_2-\Vert(\mathbf{\Phi}^{\! *}\mathbf{e}^{\prime})_{(S\cup S^{n-1})\backslash \Delta S}\Vert_2\nonumber\\
&\qquad-\Vert((\mathbf{\Phi}^{\! *}\mathbf{\Phi}-\mathbf{I})(\mathbf{x}_S-\mathbf{x}^{n-1}))_{(S\cup S^{n-1})\backslash \Delta S}\Vert_2.\label{eq:cosamp_first_12}
\end{align}
Combining (\ref{eq:cosamp_first_1}), (\ref{eq:cosamp_first_11}) and (\ref{eq:cosamp_first_12}), we have
\begin{align}
&\quad\;\Vert(\mathbf{x}_S-\mathbf{x}^{n-1})_{\overline{\Delta S}}\Vert_2\nonumber\\
&\le\Vert((\mathbf{\Phi}^{\! *}\mathbf{\Phi}-\mathbf{I})(\mathbf{x}_S-\mathbf{x}^{n-1}))_{\Delta S\backslash(S\cup S^{n-1})}\Vert_2\nonumber\\
&\quad + \Vert((\mathbf{\Phi}^{\! *}\mathbf{\Phi}-\mathbf{I})(\mathbf{x}_S-\mathbf{x}^{n-1}))_{(S\cup S^{n-1})\backslash\Delta S}\Vert_2\nonumber\\
&\quad + \Vert(\mathbf{\Phi}^{\! *}\mathbf{e}^{\prime})_{\Delta S\backslash (S\cup S^{n-1})}\Vert_2+\Vert(\mathbf{\Phi}^{\! *}\mathbf{e}^{\prime})_{(S\cup S^{n-1})\backslash\Delta S}\Vert_2\nonumber\\
&\le\sqrt{2}\Vert((\mathbf{\Phi}^{\! *}\mathbf{\Phi}-\mathbf{I})(\mathbf{x}_S-\mathbf{x}^{n-1}))_{\Delta S\cup S\cup S^{n-1}}\Vert_2\nonumber\\
&\quad +\sqrt{2}\Vert(\mathbf{\Phi}^{\! *}\mathbf{e}^{\prime})_{ \Delta S\cup S\cup S^{n-1}}\Vert_2\label{cosamp4}\\
&\overset{(\ref{rip12}), (\ref{rip13})}{\le}\sqrt{2}\delta_{4s}\Vert\mathbf{x}_S- \mathbf{x}^{n-1}\Vert_2+\sqrt{2(1+\delta_{3s})}\Vert\mathbf{e}^{\prime}\Vert_2,\label{eq:cosamp_first_result}
\end{align}
where the inequality $(\ref{cosamp4})$ is from the Cauchy-Schwartz inequality.

From the step 2 of the $n$-th iteration, noticing  that $$\text{supp}(\mathbf{x}^{n-1})\subseteq S^{n-1}\subseteq \tilde S^n \mbox{ and } \Delta S\subseteq \tilde{S}^{n},$$
we have
\begin{eqnarray}
\Vert(\mathbf{x}_S)_{\overline{\tilde{S}^{n}}}\Vert_2=\Vert(\mathbf{x}_S-\mathbf{x}^{n-1})_{\overline{\tilde{S}^{n}}}\Vert_2
\le\Vert(\mathbf{x}_S-\mathbf{x}^{n-1})_{\overline{\Delta S}}\Vert_2.\label{eq:cosamp_first_result2}
\end{eqnarray}

Hence, Lemma \ref{lem:identification-cosamp} follows by combining (\ref{eq:cosamp_first_result}) with (\ref{eq:cosamp_first_result2}).

\bibliographystyle{IEEETran}
\bibliography{Bib}
\end{document}